\documentclass{amsart}
\usepackage[english]{babel}
\usepackage[latin1]{inputenc}
\usepackage[T1]{fontenc}
\usepackage{indentfirst,amsmath,amssymb,mathrsfs,bbold}
\usepackage[dvips]{graphicx,color}

\newcommand{\A}{\mathfrak{A}}

\newcommand{\C}{\mathscr{C}}

\newcommand{\E}{\mathscr{E}}

\newcommand{\M}{\mathscr{M}}

\newcommand{\NN}{\mathbb{N}}
\newcommand{\OO}{\mathscr{O}}

\newcommand{\RR}{\mathbb{R}}
\newcommand{\Sc}{\mathscr{S}}

\newcommand{\UU}{\mathscr{U}}
\newcommand{\ud}{\mathrm{d}}
\newcommand{\VV}{\mathscr{V}}

\newcommand{\To}{\rightarrow}

\newcommand{\Then}{\Rightarrow}
\newcommand{\sm}{\smallsetminus}

\newcommand{\gotoas}[3]{\stackrel{{#2}\rightarrow{#3}}{\longrightarrow}{#1}}

\newcommand{\ol}[1]{\overline{#1}}

\newcommand{\restr}[1]{\vert_{#1}}

\begin{document}

\newtheorem{theorem}{\mdseries\scshape Theorem}[section]
\newtheorem{lemma}[theorem]{\mdseries\scshape Lemma}
\newtheorem{proposition}[theorem]{\mdseries\scshape Proposition}
\newtheorem{corollary}[theorem]{\mdseries\scshape Corollary}
\newtheorem{scholium}[theorem]{\mdseries\scshape Scholium}
\newtheorem{definition}{\mdseries\scshape Definition}[section]
\newtheorem{remark}[definition]{\mdseries\scshape Remark}
\newtheorem{conjecture}{\mdseries\scshape Conjecture}[section]




\title{Diamond-Shaped Regions as Microcosmoi}
\author{Pedro Lauridsen Ribeiro}
\address{Departamento de Matemática Aplicada, Instituto de Matemática e Estatística, Universidade de São Paulo -- Rua do Matão 1010, 05508-090 São Paulo (SP), Brazil}
\email{plaurids@ime.usp.br}
\date{\today}


\begin{abstract}
We give a geometrically intrinsic construction of a global time function for
relatively compact diamond-shaped regions in arbitrary spacetimes. In the case of
Minkowski spacetime, the flow of diffeomorphisms associated to a suitably normalized
gradient of this time function becomes the conformal isotropy subgroup of the diamond.  
In full generality, this time function is elegantly expressed in terms of the Lorentzian 
distance function, and it has an asymptotic behavior at large absolute times similar to 
the one in Minkowski spacetime. 
\end{abstract}

\maketitle




\section{Introduction}\label{intro}

\subsection{A somewhat revised (hi)story}\label{intro:hist}

A great deal of all physics is the study of time evolution of a physical system, with
the objective to be able to predict (or at least estimate) its future behavior from 
the knowledge of the past. However, the choice of what we mean by ``time'' has always 
been considered a more philosophical issue from a \emph{physical} viewpoint, whereas 
from a \emph{mathematical} viewpoint -- that is, (typically) the study of differential 
equations --, it has been nothing more than a choice of a convenient ``flow parameter''. 
With the advent of General Relativity, such a choice was forcefully brought to the 
conceptual forefront. Namely, the choice of ``time'' \emph{must be} ultimately devoid 
of any operational meaning, in the sense that it cannot fixed by any concrete physical 
procedure\footnote{Unless one makes some additional assumptions on ``test devices'', 
which amount to an approximate idealization which is not intrinsic to the theory 
\cite{westson}.}, but the laws ruling such procedures must, however, be \emph{local} 
and \emph{independent} of such a choice -- such was the lesson put forward by Einstein. 
This automatically leads to a more sensible question: what choices make \emph{manifest} 
the predictability of physical laws? Or, from a more mathematical standpoint: which 
choices render the initial value problem for the dynamics of a physical system (at 
least locally) well posed? \\

If we impose as well the requirement that physical effects propagate locally with speed 
less than or equal to the speed of light (microcausality), we are immediately led to the 
study of hyperbolic partial differential equations in general spacetimes, started by 
Hadamard \cite{had} and Riesz \cite{riesz} in a local setting. A more detailed study of 
the (Lorentzian) geometry of spacetimes reveals that the regions where \emph{global} 
well-posedness for this class of equations holds must have the property that, roughly, 
effects that can causally reach an event $p$ from the past and another event $q$ from 
the future should be spread over a \emph{finitely extended} subregion. This was proven by 
Leray in his Princeton lectures \cite{leray} -- more precisely, the Cauchy problem for 
hyperbolic partial differential equations in a Lorentzian manifold $(\M,g)$ is well 
posed within the \emph{domain of dependence} $D(\Sigma)$ of the initial data hypersurface 
$\Sigma$ (for these and all other notions of Lorentzian geometry employed in this work, 
we refer the reader to to Subsection \ref{intro:not})
\begin{equation}\label{e1}
\begin{split}
D(\Sigma) & \doteq\{p\in\M:\forall\mbox{ inextendible causal curve }\gamma\mbox{ s.t. }\gamma(\lambda_0)=p \\ 
 & \exists!\lambda_1\mbox{ s.t. }\gamma(\lambda_1)\in\Sigma\}.
\end{split}
\end{equation}

For $(\M,g)$ \emph{causal} \cite{bernsan4}, we have that $p\leq q\in D(\Sigma)$ iff
\begin{equation}\label{e2}
J^+(p)\cap J^-(q),\,J^+(p)\cap J^-(\Sigma),\,J^+(\Sigma)\cap J^-(q)\mbox{ are \emph{compact}},
\end{equation}
that is, any dynamical information that propagates from/to $\Sigma$ with speed up to that 
of light cannot leak to infinity. Such a property is tantamount to well-posedness of Cauchy 
problems of the aforementioned kind, and due to Leray is the christening ``global hyperbolicity'' 
for this property\footnote{Actually, Leray's definition of global hyperbolicity demands that the 
space of piecewise smooth causal functions linking $p$ and $q$ is either empty or compact w.r.t. 
the compact-open topology. It turns out, however, that this property is equivalent to \eqref{e2}. 
See \cite{wald} and \cite{bernsan4} for a detailed discussion on this matter.}:

\begin{definition}\label{d1}
A spacetime $(\M,g)$ such that $\M=intD(\Sigma)$ for some achronal $\C^0$ hypersurface 
$\Sigma$ is called \emph{globally hyperbolic}, and $\Sigma$ is then said to be a
\emph{Cauchy hypersurface} for $(\M,g)$.
\end{definition}


The equivalence of Definition \ref{d1} with the validity of \eqref{e2} for all $p\leq q\in\M$ 
was proven by Geroch \cite{geroch} (with some gaps filled later by Dieckmann \cite{dieckmann}); 
moreover, in this case he constructed a \emph{continuous}, surjective function $t:\M\rightarrow
\RR$ which is \emph{strictly increasing} along any future-directed causal curve (i.e. a 
\emph{global time function}) such that $t^{-1}(\lambda)$ satisfies \eqref{e1} for all 
$\lambda\in\RR$ (i.e. $t$ is a \emph{Cauchy time function}). A question that immediately arose, 
and became one of the main ``folk theorems'' in general relativity, was: can $t$ be chosen 
\emph{smooth}? This is a highly desirable property -- for instance, it implies that $\ud t$ is 
then a future directed timelike covector. A first partial answer was proposed by Seifert \cite{seifert}, 
but his smoothing procedure is difficult to understand. A satisfactory answer to this question 
ended up coming much later, with the seminal work of Bernal and Sánchez \cite{bernsan1,bernsan2,bernsan3}. 
These papers actually work in the more general case of \emph{stably causal} spacetimes, that is,
the ones which admit a (continuous) global time function $t$, and one of their main results is
that one can always regularize $t$ so as to make it smooth and still retain the characteristic
traits of a global time function. Moreover, if $t$ has Cauchy level surfaces, the regularization
procedure proposed by Bernal and Sánchez also retains this property.\\

Of course, the procedure discussed above still retains a lot of freedom. One may try to go further 
and ask whether there are choices of Cauchy time functions which are \emph{geometrically} natural, 
that is, depend only on the \emph{intrinsic} geometry of spacetime. This question is interesting not 
only in itself, but also may lead to relevant tools for addressing the dynamics of hyperbolic partial 
differential equations in the large by purely geometrical means. This is absolutely essential for the 
analysis of Einstein's equations, as dramatically manifested in the seminal work of Christodoulou
and Klainerman on the global nonlinear stability of Minkowski spacetime \cite{ck} and the developments 
originated from their ideas. For \emph{static} spacetimes, the answer to this question is an obvious
yes -- just pick a time function $\tau$ whose flow induced by the foliation of $\M$ by the level sets
of $\tau$ is generated by a Killing vector field. In a cosmological context, this question was addressed 
in considerable generality by Andersson, Galloway and Howard \cite{angalho}. Roughly, the Lorentzian 
distance to the initial cosmological singularity (which generalizes the conformal time of 
Friedmann-Robertson-Walker spacetimes) provides us such a choice. This time function is 
generally not smooth, but twice differentiable almost everywhere.

\subsection{Aims of the present work}\label{intro:aim}

Our objective is to show that, for certain globally hyperbolic regions of general spacetimes, the 
question posed at the last paragraph of the Introduction always has a positive answer. To wit, we 
give a procedure to construct global time functions from $g$ alone for regions of the form 
$\OO_{p,q}\doteq I^+(p)\cap I^-(q)$, $p\ll q\in\M$. Such regions are hereby called \emph{diamonds}. 
If $(\M,g)$ is \emph{causally simple} (that is, $J^\pm(p)$ are closed for any $p\in\M$) and $J^+(p)
\cap J^-(q)=\ol{\OO_{p,q}}$ is \emph{compact}, then such regions are automatically globally
hyperbolic (see next Section), and can be considered as ``dynamically closed'' spacetimes 
in their own right. This viewpoint pervades the whole of the paper, and is expounded
in Section \ref{sec:2}.\\

An important aspect of our construction -- which, to begin with, was its original motivation --
is that it generalizes the flow parameter of the one-parameter subgroup of the conformal 
group of the Minkowski spacetime $\RR^{1,d-1}$ that preserves a particular diamond and whose
orbits within the latter are timelike. This result is discussed in detail in Subsection 
\ref{sec:2:mink}. Indeed, for $p$ and $q$ in general spacetimes, the diffeomorphism flow 
generated by our time function approaches this flow parameter in a rather precise fashion. 
This is shown in Section \ref{sec:3}, which forms the core of the paper. The general 
construction is performed in Subsection \ref{sec:3:large}, making substantial use of tools 
from global Lorentzian geometry, namely the Lorentzian distance function \cite{beemee} and 
its fine differentiability properties, as discussed in \cite{angalho} and also in a different 
context by Moretti \cite{moretti}. More detailed properties of our global time function for 
diamonds (particularly, the aforementioned asymptotic behavior) can be obtained when the 
latter are contained in a geodesically convex neighborhood, as done in Subsection \ref{sec:3:small}. 
A slight generalization of our construction for \emph{half-diamonds}, i.e., regions of the form 
$I^\pm(q)\cap\Sigma$ for an acausal hypersurface $\Sigma$, is presented in Section \ref{sec:4}. 
This paper concludes with some comments on possible future uses of our framework.



\subsection{Notation and nomenclature}\label{intro:not}

Here we collect all basic notions of Lorentzian geometry and fix the notation we will use 
throughout the paper -- the reader can safely skip this Subsection at a first reading and
return to it when necessary. All our manifolds $\M$ are smooth, paracompact, second countable 
and oriented. Recall that a ($d$-dimensional) \emph{pseudo-Riemannian manifold} is a pair 
$(\M,g)$, where the smooth section $g\in\Gamma^\infty(\M,\vee^2T^*\M)$ is at each $p\in\M$ a 
non-degenerate symmetric bilinear form on $T_p\M$ with $r$ negative and $d-r$ positive 
eigenvalues (by Sylvester's law of inertia, $r$ is independent of the local trivialization). 
We say then that $g$ has \emph{index} $r$, or \emph{signature} $(r,d-r)$. We also say that 
$(\M,g)$ is \emph{Lorentzian} if $r=1$, and a \emph{spacetime} if, in addition, there is a 
vector field $T\in\Gamma^\infty(\M,T\M)$ such that $g(T,T)<0$ (that is, $(\M,g)$ is 
\emph{time oriented} by $T$). We denote the Levi-Civita connection associated to $g$ simply 
by $\nabla$, since no confusion will arise from this.\\

Given a spacetime $(\M,g)$, $0\neq X(p)\in T_p\M$ is said to be \emph{causal} or \emph{non-spacelike} 
(resp. \emph{timelike}, \emph{spacelike}, \emph{null}) if $g(X(p),X(p))(p)\leq 0$ (resp. $<0$, 
$>0$, $=0$) -- these properties are said to define the \emph{causal character} of $X(p)$. A 
vector field $X$ is then said to have a certain causal character if $X(p)$ has the same causal 
character for each $p\in\M$. For example, $T$ as in the previous paragraph is timelike, and 
so on; one defines analogously the causal character of covectors and 1-forms. Likewise, a 
piecewise smooth curve $\gamma:(a,b)\To\M$ is said to be \emph{causal} or \emph{non-spacelike} 
(resp. \emph{timelike}, \emph{spacelike}, \emph{null}) if $g(\dot{\gamma}(\lambda),\dot{\gamma}
(\lambda))(\gamma(\lambda))\leq 0$ (resp. $<0$, $>0$, $=0$) for all $\lambda\in(a,b)$.
$X(p)\in T_p\M$, $\omega(p)\in T^*_p\M$ are said to be \emph{future} (resp. \emph{past}) 
\emph{directed} if $g(T(p),X(p))(p),\omega(T)(p)>0$ (resp. $<0$) -- this implies that $X(p)$ 
and $\omega(p)$ are necessarily causal. $X\in\Gamma^\infty(\M,T\M)$, $\omega\in\Gamma^\infty(\M,
T^*\M)$, $\gamma:(a,b)\To\M$ are likewise said to be future (resp. past) directed if $g(T,X),
\omega(T)>0$ (resp. $<0$) everywhere on $\M$ and $g(T,\dot(\gamma)(t))>0$ (resp. $<0$) for
all $t\in(a,b)$.\\

Given $p,q\in\M$, we say that $p$ \emph{chronologically} (resp. \emph{causally}) \emph{precedes}
$q$, written as $p\ll q$ (resp. $p\leq q$) if there is a future directed timelike (resp. causal)
curve segment $\gamma:[a,b]\To\M$ such that $\gamma(a)=p$, $\gamma(b)=q$ -- we write the reverse
relations as $q\gg p$ (resp. $q\geq p$), and say that $q$ chronologically (resp. causally) succeeds 
$p$. We now define the \emph{chronological past} (resp. \emph{future}) $I^\mp(p)$ of $p\in\M$ 
and the \emph{causal past} (resp. \emph{future}) $J^\mp(p)$ of $p\in\M$ as 
\[
I^{-/+}(p)=\{q\in\M:q\ll/\gg p\},\,J{-/+}(p)=\{q\in\M:q=p\mbox{ or }\leq/\geq p\}.
\] 
We say that $\UU\subset\M$ is \emph{achronal} (resp. \emph{acausal}) if there are no $p,q\in\UU$ 
such that $p\ll q$ (resp. $p\leq q$). Given an achronal set $\UU\subset\M$, the \emph{edge} $\dot{\UU}$ 
of $\UU$ is defined as the set of the $p\in\ol{\UU}$ such that there is an open neighborhood 
$\VV\ni p$, $q\in\VV\cap I^-(p)$, $r\in\VV\cap I^+(p)$ and a future directed causal curve 
segment $\gamma:[a,b]\To\M$ such that $\gamma(a)=q$, $\gamma(b)=r$ and $\gamma([a,b])\cap
\UU=\varnothing$. One can show \cite{oneill} that any achronal set $\UU$ such that $\dot{\UU}
\cap\UU=\varnothing$ is an embedded, locally Lipschitz submanifold of $\M$ with codimension 
one. Recalling from the Introduction the notion of \emph{Cauchy development} $D(\UU)$ of $\UU$, 
one can further show \cite{oneill} that, if in addition $\UU$ is acausal, then $D(\UU)$ is open 
and hence globally hyperbolic with Cauchy hypersurface $\UU$.\\

We say that a spacetime $(\M,g)$ is \emph{chronological} (resp. \emph{causal}) if 
$p\not\ll p$ (resp. $p\not\leq p$) for all $p\in\M$, i.e. there are no closed timelike
(resp. causal) curves in $\M$, and \emph{strongly causal} if, for all $p\in\M$, there is an
open neighborhood $\UU\ni p$ such that $\gamma((a,b))\cap\UU$ is connected for all causal
curves $\gamma:(a,b)\To\M$ -- since this property is equivalent to $q\leq r\in\UU$ implying 
that all causal curve segments from $q$ to $r$ are contained in $\UU$, we say that such a
$\UU$ is \emph{causally convex}. $(\M,g)$ is \emph{stably causal} if there is a continuous, 
non-vanishing, future directed timelike vector field $S$ such that $(\M,g-S\otimes S)$ is a 
chronological spacetime, or, equivalently \cite{wald,bernsan1}, if there is $t\in\C^\infty(\M)$ 
such that $t\circ\gamma$ is strictly increasing for all future directed causal curves $\gamma:
(a,b)\To\M$ (and, hence, $\ud t$ is a future directed timelike covector). We then say that 
such a $t$ is a \emph{global time function} on $(\M,g)$. The following chain of implications
is immediate: 
\[
(\M,g)\mbox{ strongly causal}\Then(\M,g)\mbox{ causal}\Then(\M,g)\mbox{ chronological}.
\]
Let now $t\in\C^\infty(\M)$ be a global time function, and $p\in t^{-1}(0)$. Then there is
$\epsilon>0$ and a coordinate chart $(\VV,x=(x^0,\mathbf{x}))$ around $p$ such that $x(p)=0$, 
$t\restr{\VV}=x^0$ and $\UU\doteq x^{-1}(\{(x^0,\mathbf{x}):|x^0|<\epsilon^2-|\mathbf{x}|^2\})
\subset\ol{\UU}\subset\subset\VV$. By taking $\epsilon$ sufficiently small, we can make the two 
``halves'' $\partial^\pm\UU\doteq x^{-1}(\{(x^0,\mathbf{x}):x^0=\pm(\epsilon^2-|\mathbf{x}|^2)\})$ 
of the boundary $\partial\UU$ of the ``lens-shaped'' open neighborhood $\UU$ acausal, so that 
any future directed causal curve $\gamma:(a,b)\To\M$ necessarily, if ever, enters (resp. leaves) 
$\UU$ by first crossing $\partial^-\UU$ (resp. $\partial^+\UU$), which implies that $\gamma((a,b))
\cap\UU$ is always connected since $t$ is a global time function. Therefore, 
\[
(\M,g)\mbox{ stably causal}\Then(\M,g)\mbox{ strongly causal}.
\] 
It can, last but not least, be shown \cite{beemee} that 
\[
(\M,g)\mbox{ globally hyperbolic}\Then(\M,g)\mbox{ causally simple}\Then(\M,g)\mbox{ stably causal}.
\]

\section{Structure of relatively compact diamonds}\label{sec:2}

\subsection{Generalities}\label{sec:2:gen}

For any spacetime, $J^+(p)\cap J^-(q)$ is the maximal region \emph{causally accessible} 
to any observer world-line between the events $p$ and $q$. Whenever $J^+(p)\cap J^-(q)$ is 
compact, we have that $J^+(p')\cap J^-(q')$ is also compact for all $p\leq p'\ll q'\leq q$, 
\emph{if} the latter set is \emph{closed}. The latter fact is always true for \emph{causally 
simple} spacetime, thus any relatively compact \emph{diamond} 
\begin{equation}\label{e3}
\OO_{p,q}\doteq I^+(p)\cap I^-(q)=\mathrm{int}(J^+(p)\cap J^-(q)),\,p\ll q\in\M
\end{equation}
in a causally simple spacetime $(\M,g)$ is \emph{globally hyperbolic}. It is easy to
see that the Cauchy surfaces of $\OO_{p,q}$ have edge $\partial I^+(p)\cap \partial I^-(q)$.
Notice that the edge may be empty, which is precisely the case when the Cauchy surfaces of 
$\OO_{p,q}$ are closed \emph{with respect to} $\M$ \cite{oneill} and hence compact. From 
now on, $\OO_{p,q}$ will always be \emph{relatively compact}, unless otherwise stated. 


Out of technical convenience, we may (and will) make instead the stronger assumption 
that $\ol{\OO_{p,q}}$ is also contained in an open, globally hyperbolic region of $\M$, 
which particularly on its turn allows us to drop the assumption of causal simplicity.
It is appealing to point out that the former can always be achieved in the causally 
simple case (which may be taken by the reader as a model), for by taking $p\ll 
p'\ll q'\ll q$ with $\OO_{p,q}$ relatively compact, $p'$ (resp. $q'$) arbitrarily 
close to $p$ (resp. $q$), and considering $\OO_{p',q'}$ instead of $\OO_{p,q}$.\\

\subsection{Manifold structure of the boundary}\label{sec:2:man}

The boundary of a diamond $\OO_{p,q}$ is made of three disjoint pieces:

\begin{equation}\label{eq3}
\partial\OO_{p,q}=\partial_{+}\OO_{p,q}\dot{\cup}\partial_-\OO_{p,q}\dot{\cup}\E_{p,q},
\end{equation}
where
\begin{equation}\label{eq4}
\begin{split}
\partial_{+}\OO_{p,q} &\doteq\partial I^-(q)\cap I^+(p),\\
\partial_-\OO_{p,q} &\doteq\partial I^+(p)\cap I^-(q),\\
\E_{p,q} &\doteq\partial I^+(p)\cap\partial I^-(q)
\end{split}
\end{equation}
are respectively called the \emph{future horizon}, the \emph{past horizon} and the \emph{edge} 
of $\OO_{p,q}$. Recall that $\partial I^-(q)\sm\{q\}$ (resp. $\partial I^+(p)\sm\{p\}$), being 
a future (resp. past) achronal boundary, is a locally Lipschitz hypersurface ruled by future 
(resp. past) directed null geodesics which, due to causal simplicity, have a future (resp. past) 
endpoint $q$ (resp. $p$). These null geodesics are called the \emph{null generators} of $\partial 
I^-(q)$ (resp. $\partial I^+(p)$), and are supposed to be extended towards the past (resp. future) 
as long as they remain achronal. Therefore the null generators of $\partial I^-(q)$ (resp. 
$\partial I^+(p)$) are either past (resp. future) inextendible, hence past (resp. future) complete 
by causal simplicity, or have a past (resp. future) endpoint beyond which any extension ceases to 
be achronal. Such a point is called simply a \emph{past} (resp. \emph{future}) \emph{endpoint} of 
$\partial I^-(q)$ (resp. $\partial I^+(p)$). Accordingly, the restriction of the null generators 
of $\partial I^-(q)$ (resp. $\partial I^+(p)$) to$\partial_{+}\OO_{p,q}$ (resp. $\partial_-\OO_{p,q}$) 
are called the \emph{null generators} of $\partial_{+}\OO_{p,q}$ (resp. $\partial_-\OO_{p,q}$). \\

At this point, a question that naturally arises is whether the closed, achronal set 
$\E_{p,q}$ also has a manifold structure. For $\ol{\OO_{p,q}}$ contained in a geodesically
convex open neighborhood, it turns out that $\partial I^+(p)$ and $\partial I^-(q)$
are transverse at $\E_{p,q}$, so the latter is indeed a submanifold. In the general case, 
there are two potential sources of problems:

\begin{enumerate}
\item $\E_{p,q}$ is generally just the intersection of two locally Lipschitz submanifolds, 
so it is not clear what it means for these submanifolds to be transverse;
\item There may be a (non achronal) null geodesic segment $\gamma:[a,b]\To\M$ from $p$ 
to $q$ which belongs to $\partial\OO_{p,q}$, in which case $\partial I^+(p)$ and 
$\partial I^-(q)$ cannot be transverse at $\gamma([a,b])\cap\E_{p,q}$ in any reasonable
sense.
\end{enumerate}

Recall that $\partial I^+(p)\sm\{p\}$ and $\partial I^-(q)\sm\{q\}$ can be written in
a neighborhood of each of their points as the graph of a Lipschitz function in 
suitable coordinate charts. Since transversality is a local property, we can work 
in an open domain $V\subset\RR^d$ which is the range of a coordinate chart $x:U\ni 
r\To V= x(U)$ around a point $r\in\E_{p,q}$. Suppose without loss of generality that
$p,q\not\in U$; writing the components of $x$ as $x^0,\ldots,x^{d-1}$, where the coordinate
vector field $\partial_{x^0}$ is supposed to be timelike and the coordinate vector
fields $\partial_{x^1},\ldots,\partial_{x^{d-1}}$ are supposed to be spacelike, we have that 
\[
\begin{split}
x(\partial I^+(p)\cap U) &=\{(x^0=f_{+}(\mathbf{x}),\mathbf{x}):\mathbf{x}\doteq(x^1,\ldots,x^{d-1})\in\bar{V}\},\\
x(\partial I^-(q)\cap U) &=\{(x^0=f_-(\mathbf{x}),\mathbf{x}):\mathbf{x}\in\bar{V}\},
\end{split}
\]
where $f_{+},f_-:\bar{V}\To\RR$ are Lipschitz on 
\[
\bar{V}=\{\mathbf{x}\in\RR^{d-1}:(x^0(r),\mathbf{x})\in V\}.
\]

If $\partial I^+(p)$ and $\partial I^-(q)$ are $\C^1$ at $r$, then they are transverse
at $r$ if and only if $\ud f_{+}\circ x(r)$ and $\ud f_-\circ x(r)$ are linearly
independent. At points where either $f_{+}$ or $f_-$ is not differentiable, we can use
the fact that any Lipschitz function $f$ on $\bar{V}$ is differentiable almost everywhere 
(Rademacher's theorem) and define the \emph{generalized differential} of $f$ at $r$ 
\cite{clarke} as the set
\[
\partial f(r)\doteq\mathrm{co}\{\lim\ud f(x_i):x_i\To x,x_i\not\in\Omega f\}
\]
where $\mathrm{co}$ denotes ``convex hull'' and $\Omega_f$ denotes the (null) set of 
points of $\bar{V}$ where $f$ is not partially differentiable with respect to all
coordinates. If $f$ is $\C^1$ at $r$, then one can show that $\partial f(r)=\{\ud f(r)\}$.
Using this definition, we say that $\partial I^+(p)$ and $\partial I^-(q)$ are 
\emph{transverse} at $r$ if any element of $\partial f_{+}(r)$ is linearly independent 
of any element of $\partial f_-(r)$. This definition is independent of the choice of
coordinates and functions $f_{+},f_-$ used to represent $\partial I^+(p)\cap U$ and 
$\partial I^-(q)\cap U$. If $\partial I^+(p)$ and $\partial I^-(q)$ are transverse in 
an open neighborhood of $r$ in $\E_{p,q}$, then we can invoke the same argument
used in the $\C^1$ case to reduce the problem of proving that the the transverse
intersection of two submanifolds is also a submanifold to an application of the
inverse function theorem \cite{hirsch}. At this point, we employ instead of the following

\begin{theorem}[Clarke's Inverse Function Theorem \cite{clarke}] 
Let $U\subset\RR^n$ open, $F:U\To\RR^n$ Lipschitz, $x_0\in U$. Define the generalized
Jacobian of $F$ at $x_0$ as the set of $n\times n$ matrices
\[
\partial F(x_0)\doteq\mathrm{co}\{\lim\ud F(x_i):x_i\To x,x_i\not\in\Omega F\},
\]
where $\Omega_F$ denotes the (null) set of points of $\bar{V}$ where the components of
$F$ are not all partially differentiable with respect to all coordinates. If all elements
of $\partial F(x_0)$ are nonsingular, then there are open neighborhoods $V\subset U,W$ 
of $x_0$ and a Lipschitz function $G:W\To\RR^n$ such that $G\circ F(v)=v$ for all $v\in V$ 
and $F\circ G(w)=w$ for all $w\in W$.\qed
\end{theorem}

To summarize, we have proven the

%

\begin{corollary}
Let $f_{+},f_-$ be transverse at $\mathbf{x}_0$. Then there is a neighborhood $\bar{W}
\subset\bar{V}$ of $\mathbf{x}_0$ such that the set
\[
x(\E_{p,q})\cap ((\RR\times\bar{W})\cap V) =\{(x^0=f_{+}(\mathbf{x})=f_-(\mathbf{x}),
\mathbf{x}):\mathbf{x}\in\bar{W}\}
\]
is the graph of a Lipschitz function.\qed
\end{corollary}


%

Using the local graph description of Lipschitz submanifolds outlined above, we
can define the \emph{conormal cone} at a point $r$ of $\partial I^+(p)\sm\{p\}$
or $\partial I^-(q)\sm\{q\}$ as the convex hull of the limits of conormal 
directions along sequences of differentiable points converging to $r$.
Since any point of an achronal boundary which is not an endpoint of more than 
one geodesic is differentiable, we see that the conormal cone of $\partial I^+(p)\sm\{p\}$
(resp. $\partial I^-(q)\sm\{q\}$) at each point is given by the convex hull of
covectors of the form $g(\dot{\gamma},.)$, where $\gamma$ is a future (resp. past)
directed null generator of $\partial I^+(p)\sm\{p\}$ (resp. $\partial I^-(q)\sm\{q\}$)
with respect to some affine parametrization. 


Now we are in a position to be more precise about the second problem, for it describes
exactly what happens at points where the intersection of $\partial I^+(p)$ with $\partial 
I^-(q)$ is not transverse even in the above generalized sense. It is clear that, 
if $\gamma:[a_-,b_{+}]\To\M$ is any null geodesic segment connecting $p\ll q$, there must 
be $a_-,<a_{+},b_-<b_{+}$ such that $p_{+}\doteq\gamma(a_{+})$ is a future endpoint of 
$\partial I^+(p)$ and $q_-\doteq\gamma(b_-)$ is a past endpoint of $\partial 
I^-(q)$. 

\begin{lemma}
Let $\gamma$ be as above and suppose that it is completely contained in $\partial\OO_{p,q}$. 
Then $p_{+}$ and $q_-$ must belong to $\E_{p,q}$, in which case we must have $b_-\leq a_{+}$.
Moreover, if $q_-\neq p_{+}$, then the whole null geodesic segment $\gamma\restr{[a_{+},b_-]}$ 
must belong to $\E_{p,q}$.
\end{lemma}
\begin{proof}
Indeed, if either $p_{+}$ or $q_-$ do not belong to $\E_{p,q}$, then $\gamma$ must leave 
$\partial\OO_{p,q}$ at one of these points, in contradiction with the hypothesis.
If $a_{+}<b_-$, $\gamma\restr{[a_-,b_-]}$ (resp. $\gamma\restr{[a_{+},b_{+}]}$) would 
be a non achronal null geodesic segment connecting $p$ (resp. $p_{+}\in\partial I^-(q)$) to 
$q_-\in\partial I^+(p)$ (resp. $q$), violating the achronality of both $\partial I^+(p)$ and 
$\partial I^-(q)$. Suppose now that there are $a_{+}<c<d<b_-$ such that $\gamma(c),\gamma(d)
\in\E_{p,q}$ and $\gamma((c,d))\cap\E_{p,q}=\varnothing$. Then by hypothesis either $\gamma((c,d))
\subset \partial_-\OO_{p,q}$ or $\gamma((c,d)) \subset\partial_{+}\OO_{p,q}$. The first (resp. 
second) case can take place only if $\gamma$ leaves $\partial_{+}\OO_{p,q}$ (resp. 
$\partial_-\OO_{p,q}$) at $\gamma(c)$ and returns at $\gamma(d)$, which implies that $\gamma(d)$ 
is a past (resp. future) endpoint of $\partial_{+}\OO_{p,q}$ (resp. $\partial_-\OO_{p,q}$). 
However, this implies that $\gamma\restr{[c,b_{+}]}$ (resp. $\gamma\restr{[a_-,d]}$) is a non 
achronal null geodesic connecting $\gamma(c)\in\partial I^-(q)$ (resp. $p$) to $q$ (resp. 
$\gamma(d)\in\partial I^+(p)$), in contradiction with the achronality of $\partial I^-(q)$ 
(resp. $\partial I^+(p)$).
\end{proof}

Conversely, we have the following:

\begin{lemma}
Let $p\ll q$ belong to a causally simple spacetime $(\M,g)$. If $\E_{p,q}$ contains a
null geodesic segment $\gamma:[a,b]\To\M$, then $\gamma$ can be extended to a (non
achronal) null geodesic segment from $p$ to $q$ which is completely contained in 
$\partial\OO_{p,q}$ and whose intersection with $\E_{p,q}$ is a null geodesic segment 
whose past (resp. future) endpoint is a future (resp. past) endpoint of $\partial 
I^+(p)$ (resp. $\partial I^-(q)$) and whose interior intersects no null generator
of either $\partial I^+(p)$ or $\partial I^-(q)$ other than $\gamma$ itself.
\end{lemma}
\begin{proof}
If $\gamma$ is as in the hypothesis, then it can be past (resp. future) extended beyond 
$\gamma(a)$ (resp. $\gamma(b)$) up to $p$ (resp. $q$). Let $a_-$ and $b_{+}$ the values 
of the affine parameter of $\gamma$ such that $p=\gamma(a_-)$ and $q=\gamma(b_{+})$. It 
follows that $\gamma\restr{[a_-,b]}$ (resp. $\gamma\restr{[a,b_{+}]}$) is part of a null 
generator of $\partial_-\OO_{p,q}$ (resp. $\partial_{+}\OO_{p,q}$). However, the fully 
extended $\gamma$ ranging from $a_-$ to $p_{+}$ cannot be achronal, for $p\ll q$. This 
implies that there must be $a_-<a_{+}\leq a$, $b\leq b_-<b_{+}$ such that $\gamma(b_-)$ 
(resp. $\gamma(a_{+})$) is a future (resp. past) endpoint of the null generator $\gamma
\restr{[a_-,b_-]}$ (resp. $\gamma\restr{[a_{+},b_{+}]}$) of $\partial_-\OO_{p,q}$ (resp. 
$\partial_{+}\OO_{p,q}$), and $\gamma\restr{[a_-,b_{+}]}$ is the maximal extension of 
$\gamma$ that still belongs to $\E_{p,q}$. The very existence of such a $\gamma$ entails, 
in addition, that $\gamma(a_{+})$ (resp. $\gamma(b_-)$) is \emph{not} a future (resp. 
past) endpoint of $\partial_-\OO_{p,q}$ (resp. $\partial_{+}\OO_{p,q}$). This implies 
that no other generator of either $\partial_-\OO_{p,q}$ or $\partial_{+}\OO_{p,q}$ 
can cross $\gamma\restr{(a_-,b_{+})}$.
\end{proof}

The above results give us a picture about the most pathological points of $\E_{p,q}$. 
Now we show that, for ``most'' choices of pairs of points $p\ll q$, these points are
rather rare. 

\begin{proposition}
Let $(\M,g)$ be globally hyperbolic, $q\in\M$. Then the set of points $p\ll q$ such 
that there are only a finite number of future directed nonspacelike geodesics from $p$
to $q$ is residual in $I^-(q)$.
\end{proposition}
\begin{proof}[Sketch]
By Sard's theorem, the set of points $p\ll q$ which are not conjugate along any
geodesic are residual in $\M$, so let us concentrate only on these. The argument
is completed along the lines of Proposition 10.42 of \cite{beemee}, pp. 364.
\end{proof}

The results of this Subsection can be summarized by the following

\begin{theorem}
Let $(\M,g)$ be globally hyperbolic, $q\in\M$. Then, for a residual set of points
$p\in I^-(q)$, the edge $\E_{p,q}$ of $\OO_{p,q}$ is a compact Lipschitz submanifold
of $\M$, up to a finite set of points and a finite set of null geodesic segments.\qed
\end{theorem}

It is easy to see from \eqref{e1} that any Cauchy surface $\Sigma$ for $\OO_{p,q}$ has
$\E_{p,q}$ as its boundary \emph{but}, although $\OO_{p,q}$ is of the form $D(\Sigma)$ 
as a spacetime \emph{by itself}, it is \emph{not} of this form as a \emph{proper} 
region of $\M$, for generally $\partial_{+}\OO_{p,q}$ (resp. $\partial_-\OO_{p,q}$) 
possesses achronal geodesic segments with future (resp. past) endpoints, that is, 
there are points which belong to null generators which cease to be achronal before 
reaching $\E_{p,q}$. Hence, we 
have that $\mathrm{int}D(\Sc_{p,q})\supsetneqq\OO_{p,q}$ -- equality occurs if and only 
if the null geodesic segments which generate $\partial_{+}\OO_{p,q}$ and $\partial_-
\OO_{p,q}$ are achronal, as for example in the case that $\ol{\OO_{p,q}}$ is contained 
in some geodesically convex neighborhood.

\subsection{A remark on conformal isotropy groups of diamonds in Minkowski spacetime}\label{sec:2:mink}

Let us consider now the simplest case, that of diamonds in Minkowski spacetime $\RR^{1,d-1}$. 
Each $\OO_{p,q}$ possesses a one-parameter isotropy subgroup of the conformal group $SO(2,d)$ 
of $\RR^{1,d-1}$, given in standard Cartesian coordinates $x=(x^0=t,\mathbf{x})$, $\mathbf{x}=
(x^1,\ldots,x^{d-1}))$, $|\mathbf{x}|=r$ by
\begin{equation}\label{e4}
\begin{split}
x\mapsto u^\lambda_{p,q}(x) & = K^{-1}_{p,q}(e^\lambda K_{p,q}(x)),\mbox{ where}\\
K_{p,q}(x) & \doteq K\left(\Lambda_{p,q}\left(x-\frac{x(p)+x(q)}{2}\right)\right),
\end{split}
\end{equation}
where $\Lambda_{p,q}$ is the Lorentz boost around the origin making the direction 
$\overrightarrow{x(p)x(q)}$ parallel to the $x^0$ axis, $K(x)=I(x-(1,\mathbf{0}))-
\left(\frac{1}{2},\mathbf{0}\right)$ and $I(x)=-\frac{x}{\eta(x,x)}$ is the relativistic 
ray inversion map. That is, the discrete conformal transformation $K_{p,q}$ is a 
diffeomorphism of $\OO_{p,q}$ onto the forward light cone $I^+(0)$ which maps the 
(multiplicative) subgroup of dilations (written additively by putting the dilation in 
the exponential form $e^\lambda$) onto $\lambda\mapsto u^\lambda_{p,q}$ by conjugation. 
Since conformal transformations preserve the causal structure of strongly causal
 spacetimes 
and the orbits of $\lambda
\mapsto e^\lambda\cdot$ in $I^+(0)$ are everywhere timelike, so are the orbits of 
$\lambda\mapsto u^\lambda_{p,q}$ in $\OO_{p,q}$.

Another description of $\lambda\mapsto u^\lambda_{p,q}$ can be given as follows.
Consider the radial null coordinates $x^-=t-r$ (\emph{retarded time}) and $x^-=t+r$ 
(\emph{advanced time}), whose level sets are respectively the forward and backward light cones 
$\partial I^+((t,\mathbf{0}))$ and $\partial I^-((t,\mathbf{0}))$. In the case 
that $p=p_0\doteq (-1,\mathbf{0})$ and $q=q_0\doteq (1,\mathbf{0})$, $u^\lambda_{p_0,q_0}$ 
acts \emph{only on the coordinates} $x^\pm$, and \eqref{e4} can be simplified to 
\[
x^{\pm}(u^\lambda_{p_0,q_0}(p))\doteq\frac{(1+x^{\pm})-e^{-\lambda}(1-x^{\pm})}{(1+x^{\pm})+e^{-\lambda}
(1-x^{\pm})},\,\lambda\in\RR.
\] 
This leads to a quite appealing, intrinsic geometric 
characterization of the parameter $\lambda$, which we state already in coordinate-independent 
form:

\begin{theorem}\label{t1}
If $\lambda$ is the parameter of $u^\lambda_{p,q}$, seen as a \emph{global time function}
on $\OO_{p,q}$, then \begin{equation}\label{e5}\lambda(r)=\frac{1}{2}\log\left[
\frac{d_{\eta}(p,r)^2}{d_{\eta}(r,q)^2}\right]=\frac{1}{d}\log\left[\frac{\mbox{Vol}
\OO_{p,r}}{\mbox{Vol}\OO_{r,q}}\right],\end{equation} where $d^2_{\eta}(r,q)$ is the 
square of the \emph{Lorentzian distance} associated to the Minkowski spacetime metric 
$\eta$.
\end{theorem}
\begin{proof}
Due to the rotational symmetry of $u^\lambda_{p_0,q_0}$, we can restrict our 
considerations to the $x^0-x^1$ plane. In this case, let us write $x^\mu(p)=(0,x^1,0,\ldots,0)$, 
$x^1\in[-1,1]$, whence it follows that 
\[
x^{\pm}(u^\lambda_{p_0,q_0}(p))=\frac{(1\pm x^1)-e^{-\lambda}(1\mp x^1)}{(1\pm x^1)+e^{-\lambda}(1\mp x^1)},
\] 
and thus 
\begin{align*}
x^0(u^\lambda_{p_0,q_0}(p)) & =\frac{1}{2}(x^+(u^\lambda_{p_0,q_0}(p))+x^-(u^\lambda_{p_0,q_0}(p)))\\ 
 & =\frac{(1-(x^1)^2)(1-e^{-2\lambda})}{(1-(x^1)^2)(1+e^{-2\lambda})+2e^{-\lambda}(1+(x^1)^2)}
\end{align*}
and 
\begin{align*}
x^1(u^\lambda_{p_0,q_0}(p)) & =\frac{1}{2}(x^+(u^\lambda_{p_0,q_0}(p))-x^-(u^\lambda_{p_0,q_0}(p)))\\
 & = \frac{4x^1e^{-\lambda}}{(1-(x^1)^2)(1+e^{-2\lambda})+2e^{-\lambda}(1+(x^1)^2)}.
\end{align*}

Let $q=(t,\mathbf{0})$, $t\in(-1,1)$. The diamond $\OO_{p_0,q}$ is a translation
of $\frac{1+t}{2}\OO_{p_0,q_0}$, and the diamond $\OO_{q,q_0}$, a translation
of $\frac{1-t}{2}\OO_{p_0,q_0}$ -- hence, we have $|\OO_{p_0,q}|=\left(\frac{1+t}{2}
\right)^{d-1}|\OO_{p_0,q_0}|$ and $|\OO_{q,q_0}|=\left(\frac{1-t}{2}\right)^{d-1}
|\OO_{p_0,q_0}|$. More generally, if $q=(t,r,0,$\\$\ldots,0)$, there exists a Lorentz 
boost in the $x^0-x^1$ plane around $p_0$ which takes $\OO_{p_0,q}$ to $\OO_{p_0,q^+}$, 
where $q^+=(((1+t)^2-r^2)^{\frac{1}{2}}-1,\mathbf{0})$, and a Lorentz boost in the $x^0-x^1$ 
plane around $q_0$ which takes $\OO_{p_0,q}$ to $\OO_{q^-,q_0}$, where $q^-=(1-((1-t)^2-
r^2)^{\frac{1}{2}},\mathbf{0})$. As Lorentz transformations preserve volume, we have 
\[
|\OO_{p_0,q}|=\left(\frac{((1+t)^2-r^2)^{\frac{1}{2}}}{2}\right)^d|\OO_{p_0,q_0}|
\] 
and 
\[
|\OO_{q,q_0}|=\left(\frac{((1-t)^2-r^2)^{\frac{1}{2}}}{2}\right)^d|\OO_{p_0,q_0}|.
\] 
Finally, taking $q=q(\lambda)=u^\lambda_{p_0,q_0}(p)$, it follows that 
\[
(1+x^0(q(\lambda)))^2-x^1(q(\lambda))^2=\frac{4(1-(x^1)^2)}{(1-(x^1)^2)(1+e^{-2\lambda})+2e^{-\lambda}(1+(x^1)^2)}
\] 
and 
\[(1-x^0(q(\lambda)))^2-x^1(q(\lambda))^2=\frac{4e^{-2\lambda}(1-(x^1)^2)}{(1-(x^1)^2)(1+e^{-2\lambda})+2e^{-\lambda}
(1+(x^1)^2)},
\] 
and hence formula \eqref{e5}. The second identity follows from the fact that $|\OO_{p,q}|=
\frac{1}{2^{d-1}d(d-1)}\mbox{\upshape Vol}S^{d-2}d_\eta(p,q)^d$.
\end{proof}

\section{Cosmological time functions in general diamonds}\label{sec:3}

\subsection{The structure in the large}\label{sec:3:large}

Let us denote by 
\begin{equation}\label{e6}
\log\frac{d_g(p,\cdot)}{d_g(\cdot,q)}\doteq\lambda^g_{p,q}:\OO_{p,q}\rightarrow\RR
\end{equation}
the function given by the first identity in \eqref{e5} with the Lorentzian distance 
$d_\eta$ associated to the Minkowski metric $\eta$ replaced by the one associated to $g$,
henceforth called $d_g$ (see below). 

\begin{remark}\label{s3r1}
In Minkowski spacetime, the second identity in \eqref{e5} equates $\lambda^g_{p,q}$, up to a 
constant factor, to the global time function originally built in \cite{geroch}. For diamonds in 
\emph{general}, causally simple spacetimes, the second identity no longer holds, due to curvature 
effects -- the intuitive reason is that the ``packing'' number of small diamonds inside a larger 
one need not grow linearly with the volume of the latter. This heuristic argument can be made 
rigorous by employing semi-Riemannian volume comparison estimates \cite{ehrsan}. 
\end{remark}

Although $\lambda^g_{p,q}$ suggests itself as a natural choice of global time 
function for $\OO_{p,q}$ for general spacetimes, the remark at the beginning of Section 
\ref{sec:2} raises concerns about whether the level sets of $\lambda^g_{p,q}$ are 
Cauchy surfaces or not. Before proceeding any further, let us recall the general 
definition of $d_g$:

\begin{equation}\label{e7}
d_g(p,q)\doteq\left\{\begin{array}{cl}{{\mbox{\normalsize sup}}\atop{\mbox{\scriptsize 
$\gamma\in\Omega_{p,q}$}}}\sum^k_{i=1}\int^{\lambda_i}_{\lambda_{i-1}}\sqrt{-g(\dot{\gamma}_i
(\lambda),\dot{\gamma}_i(\lambda))}d\lambda & \mbox{if }p\leq q \\ 0 & \mbox{otherwise}
\end{array}\right.,
\end{equation}
where $\Omega_{p,q}$ is the set of future directed, piecewise $\C^\infty$ causal 
curves from $p$ to $q$. $d_g$ as defined in \eqref{e7} enjoys the fundamental 
\emph{reverse triangular inequality}
\begin{equation}\label{e8}
d_{\bar{g}}(p,q)\geq d_{\bar{g}}(p,r)+d_{\bar{g}}(r,q).
\end{equation}

With \eqref{e7} at our disposal, we can prove the following
\begin{proposition}\label{p1}
The level sets of 
\[
\lambda^g_{p,q}=\frac{1}{2}\log\frac{d_g(p,\cdot)^2}{d_g(\cdot,q)^2}
\] 
are acausal.
\end{proposition}
\begin{proof}
Since $\frac{1}{2}\log$ is injective and strictly monotonically increasing,
it suffices to establish the claim for the ratio $\frac{d_g(p,\cdot)}{d_g(\cdot,q)}$.
It follows from \eqref{e8} that, if $\gamma:[0,1]\To\M$ is a future directed 
causal curve segment in $\OO_{p,q}$, then $\lambda\mapsto d_g(\gamma(\lambda),q)$ 
(resp. $\lambda\mapsto d_g(p,\gamma(\lambda))$) is a function bounded by $d_g(p,q)$ 
($<+\infty$ by virtue of the compactness of $\ol{\OO_{p,q}}$ and the definition of 
Lorentzian distance) and \emph{strictly} decreasing (resp. increasing) in $\lambda$. 
To see the latter fact, recall that any \emph{maximal} causal curve (i.e. whose 
arc length between any two of its points is equal to the Lorentzian distance) is 
necessarily a smooth geodesic, up to reparametrization \cite{beemee}. Therefore, 
even if $\gamma$ is an achronal null geodesic segment (and thus $d_g(\gamma(0),
\gamma(1))=0$), the fact that $\gamma(0),\gamma(1)\in\OO_{p,q}$ implies that any 
future directed causal curve segment linking either $p$ to $\gamma(1)$ or 
$\gamma(0)$ to $q$, and containing $\gamma$, will either explicitly have a larger 
Lorentzian arc length or be a \emph{broken} causal curve segment, which for no 
reparametrization can be made a smooth geodesic. 

The above argument shows particularly that the level sets of $d_g(.,q)$ and 
$d_g(p,.)$ are \emph{acausal} at nonzero values. To reach the conclusion of 
the theorem, we argue by \emph{reductio ad absurdum} and assume that 
\[
\frac{d_g(p,\gamma(0))}{d_g(\gamma(0),q)}=\frac{d_g(p,\gamma(1))}{d_g(\gamma(1),q)}.
\] 
This implies that 
\[
d_g(p,\gamma(0))=\frac{d_g(\gamma(0),q)}{d_g(\gamma(1),q)}d_g(p,\gamma(1)).
\] 
However, we have seen that $d_g(p,\gamma(0))<d_g(p,\gamma(1))$ and 
$d_g(\gamma(0),q)>d_g(\gamma(1),q)$ -- the second inequality implies that 
$d_g(p,\gamma(0))>d_g(p,\gamma(1))$, in contradiction with the first inequality.
\end{proof}

As for Cauchy property of the level sets proper, notice that, due to strong causality, 
any inextendible causal curve $\gamma:(0,1)\To\OO_{p,q}$ has a past endpoint $\gamma_-$ 
in $I^-(q)\cap\partial I^+(p)$ and a future endpoint $\gamma_+$ in $I^+(p)\cap\partial I^-(q)$. 
In this case, we have 
\[
d_g(p,\gamma_-)=d_g(\gamma_+,q)=0,\mbox{ and }d_g(\gamma_-,q),d_g(p,\gamma_+)\neq 0.
\] 
Continuity of $\gamma$ will then assure that $\frac{d_g(p,\gamma(t))}{d_g(\gamma(t),q)}$ 
assume all possible values in $(0,+\infty)$, \emph{provided} that the latter ratio is 
continuous \emph{as well}, which on its turn is a consequence of the fact that $\ol{\OO_{p,q}}$ 
is assumed to be contained in a globally hyperbolic region of $(\M,g)$, in which $d_g$ is 
thus jointly continuous and finite valued (Lemma 4.5, page 140 in \cite{beemee}). This yields
\begin{corollary}\label{c1}
$\lambda^g_{p,q}$ is continuous and the level sets of $\lambda^g_{p,q}:\OO_{p,q}\To\RR$ are 
acausal Cauchy surfaces.\qed
\end{corollary}

Actually, one can say even more. Under our hypotheses on $p,q$, we have that the part of
the future non-spacelike cut locus of $p$ and the part of the past non-spacelike cut locus
of $q$ within $\ol{\OO}_{p,q}$ are closed \cite{beemee}. Moreover, by definition any open 
neighborhood of a future non-spacelike cut point of $p$ contains a point which can be
connected to $p$ by a unique maximal past directed causal geodesic segment; likewise, 
any open neighborhood of a past non-spacelike cut point of $q$ contains a point which 
can be connected to $q$ by a unique maximal future directed causal geodesic segment.
In other words, 
\begin{lemma}
Given $p\ll q$ contained in a globally hyperbolic region of $(\M,g)$, the complement of 
the union of the future non-spacelike cut locus of $q$ and the past non-spacelike cut locus
of $q$ in $\ol{\OO}_{p,q}$ is open and dense in the relative topology.\qed
\end{lemma}

This also implies that $\lambda^g_{p,q}$ is smooth almost everywhere, for the set of 
non-spacelike cut points of $p$ and $q$ within $\ol{\OO}_{p,q}$ is then compact and 
nowhere dense. 


\begin{proposition}\label{p2}
$d_g(p,.)$ and $d_g(.,q)$ are \emph{semi-convex}, i.e., for any $r\in\OO_{p,q}$ there exists a 
neighborhood  $\UU\ni r$, a local chart $x:\UU\To\RR^d$ and $f\in\C^\infty(\UU)$ such that 
$(d_g(p,.)\restr{\UU}+f)\circ x^{-1}$ and $(d_g(.,q)\restr{\UU}+f)\circ x^{-1}$ are \emph{convex} 
in $x(\UU)$.
\end{proposition}
\begin{proof} (Sketch; for more details, see \cite{angalho} and references therein)
We will present the argument only for $d_g(p,.)$, for the argument for $d_g(.,q)$ is 
analogous. There is a future-directed maximizing timelike geodesic segment, say 
$\gamma_p:[0,d_g(p,r)]\To\M$, from $p$ to $r\in\OO_{p,q}$ 
(that is, $\gamma_p(0)=p$, $\gamma_p(d_g(p,r))=r$ and $d_g(\gamma_p(\lambda),\gamma_p
(\lambda'))=\lambda'-\lambda$ for all $\lambda'>\lambda$). Since $\gamma_p$ is maximizing, 
the segment $\gamma_p(\epsilon,d_g(p,r))$ is free of cut points for any $0<epsilon<d_g(p,r)$
(from now on fixed) and hence there is an open neighborhood $\UU_\epsilon$ of $\gamma_p
([\epsilon,d_g(p,r)])$ in $\OO_{p,q}$ where $d_g(\gamma_p(\epsilon),.)$ is smooth. We use
$\gamma_p(\epsilon)$ instead of $p$ in the first entry of $d_g$ because, whereas $\gamma_p
\restr{[\epsilon,d_g(p,r)]}$ the unique maximizing timelike geodesic segment linking 
$\gamma_p(\epsilon)$ to $r$, it is not necessarily true that $\gamma_p$ is the unique 
maximizing timelike geodesic segment linking $p$ to $r$. Nevertheless, we still have 
$d_g(p,r)=d_g(p,\gamma_p(\epsilon))+d_g(\gamma_p(\epsilon),r)$, of course, but not 
necessarily $d_g(p,\gamma_p(d_g(p,r)+\delta))=d_g(p,\gamma_p(\epsilon))+d_g(\gamma_p
(\epsilon),\gamma_p(d_g(p,r)+\delta)$, no matter how small $0<\delta$ is.

Particularly, we have $g^{-1}(\ud d_g(\gamma_p(\epsilon),.),\ud d_g(\gamma_p(\epsilon),.))=-1$ 
everywhere in $\UU_\epsilon$ and the covariant Hessian $\mbox{Hess}d_g(p,.)(X,Y)=\nabla_X
\nabla_Yd_g(\gamma_p(\epsilon),.)$ exists everywhere in $\UU_\epsilon$ and defines the second 
fundamental form (whose associated linear operator is the Weingarten map) of the level 
hypersurfaces of $d_g(\gamma_g(\epsilon),.)$ in $\UU_\epsilon$. Consider particularly an open 
normal coordinate neighborhood $\VV\subset\UU_\epsilon$ of $r$. Then $\mbox{Hess}d_g(p,.)$ 
as a quadratic form satisfies two-sided bounds in $\VV$ in terms of existing two-sided bounds 
on the sectional curvature of 2-planes containing $\dot{\gamma}_p$ \cite{anderho}.
\end{proof}

More precisely, there exist $\phi_{p,r},\phi_{r,q}\in\C^\infty(\UU)$ such that $d_g(p,r)=
\phi_{p,r}(r)$ and $d_g(r,q)=\phi_{r,q}(r)$, $d_g(p,.)\geq\phi_{p,r}$ and $d_g(.,q)\geq
\phi_{r,q}$ in $\UU$ and the Hessians $D^2\phi_{p,r}$ and $D^2\phi_{r,q}$ are such that 
$D^2\phi_{p,r}(r)-c_{p,r}\mathbb{1}$ and $D^2\phi_{r,q}(r)-c_{r,q}\mathbb{1}$ are positive 
semidefinite matrices for $c_{p,r},c_{r,q}\in\RR$, which not only implies semi-convexity 
in the sense of Proposition \ref{p2} \cite{angalho}, but also guarantees that the given 
definition is independent of coordinates. In these circumstances, we can invoke the 
classical result of Aleksandrov \cite{evangar}, which tells us that a convex function is 
not only locally Lipschitz, but is also twice differentiable almost everywhere with 
respect to Lebesgue measure. Such a result obviously extends to semi-convex functions. 
As the restriction of $\mu_g$ to normal neighborhoods is absolutely continuous with 
respect to Lebesgue measure, we thus obtain the following
\begin{proposition}\label{p3}
$d_g(p,.)^2$ and $d_g(.,q)^2$ are locally Lipschitz and twice differentiable almost 
everywhere in $\OO_{p,q}$.\qed
\end{proposition}

Now we want to prove a simple but important approximation property for $\lambda^g_{p,q}$,
already suggested by the proof of Proposition \ref{p2}.

\begin{proposition}\label{p4}
Let $(p_n)_{n\in\NN}$, $(q_n)_{n\in\NN}$ be two sequences of points in $\M$ such that
$p\leq p_n\ll q_n\leq q$, $p_n\gotoas{n}{\infty} p$, $q_n\gotoas{n}{\infty} q$. Then,
for any $K\subset\OO_{p,q}$ compact, we have (possibly after passing to subsequences
so as to guarantee that $K\subset\OO_{p_n,q_n}$ for all $p_n,q_n$ in the respective 
subsequences) that $\sup_K|\lambda^g_{p_n,q_n}-\lambda^g_{p,q}|\gotoas{n}{\infty} 0$.
\end{proposition}
\begin{proof}
Obviously, we have $\OO_{p_n,q_n}\subset\OO_{p,q}$ for all $n$. From now on, we tacitly
assume without loss of generality that $K\subset\OO_{p_n,q_n}$ for all $n$. The assertion
then follows from the continuity of $d_g$.
\end{proof}

An important peculiarity of $\lambda^g_{p,q}$ in Minkowski spacetime which also survives 
to a great extent in the general case is that a timelike geodesic segment linking $p$ to
$q$ which maximizes the Lorentzian distance between these two points will be, up to (an
inevitable) reparametrization, a complete orbit of the diffeomorphism flow associated to
the foliation induced by $\lambda^g_{p,q}$. 

\begin{proposition}\label{p5}
Any maximal, unit-speed future directed timelike geodesic segment $\gamma:[0,d_g(p,q)]$ 
between $p\ll q$ realizes the Lorentzian distance to the level sets of $\lambda^g_{p,q}$.
That is, for each $\lambda\in\RR$, if $t_\lambda\in(0,d_g(p,q))$ is such that $\gamma
((0,d_g(p,q)))\cap(\lambda^g_{p,q})^{-1}(\lambda)=\{\gamma(t_\lambda)\}$, then
\[
d_g(\gamma(t'),(\lambda^g_{p,q})^{-1}(\lambda))\doteq\sup_{r''\in\Sigma_\lambda}d_g
(\gamma(t'),r'')=d_g(\gamma(t'),\gamma(t_\lambda)),\,\forall t'\leq t_\lambda,
\]
\[
d_g((\lambda^g_{p,q})^{-1}(\lambda),\gamma(t'))\doteq\sup_{r''\in\Sigma_\lambda}d_g(r'',
\gamma(t'))=d_g(\gamma(t_\lambda),\gamma(t')),\,\forall t'\geq t_\lambda.
\]
\end{proposition}
\begin{proof}
We will just prove the claim for $t'\geq t_\lambda$, the other case being analogous. Let 
$r\in\Sigma_\lambda$ realize the Lorentzian distance between $\Sigma_\lambda\doteq
(\lambda^g_{p,q})^{-1}(\lambda)$ and $r'\doteq\gamma(t')$, i.e. $d_g(r,r')=d_g
(\Sigma_\lambda,r')$, and $r_\lambda=\gamma(t_\lambda)$. The existence of $r$ follows 
from the continuity of $d_g$ within $\OO_{p,q}$ and the fact that all past directed, 
past inextendible causal curves in $\M$ issuing from $r'$ must cross $\Sigma_\lambda$ 
in a compact subset thereof before leaving $\OO_{p,q}$ since the former is a Cauchy 
hypersurface w.r.t. the latter. We have the following six facts: from (a) the 
maximality of $\gamma$, it follows that 
\begin{itemize}
\item[] (a1)\quad $d_g(p,r_\lambda)+d_g(r_\lambda,r')=d_g(p,r')$;
\item[] (a2)\quad $d_g(r_\lambda,r')+d_g(r',q)=d_g(r',q)$;
\end{itemize}
from (b) the definition of $\Sigma_\lambda$,
\begin{itemize}
\item[] (b1)\quad $d_g(p,r)=e^\lambda d_g(r,q)$;
\item[] (b2)\quad $d_g(p,r_\lambda)=e^\lambda d_g(r_\lambda,q)$;
\end{itemize}
finally, from (c) the reverse triangular inequality for $d_g$,
\begin{itemize}
\item[] (c1)\quad $d_g(p,r)+d_g(r,r')\leq d_g(p,r')$;
\item[] (c2)\quad $d_g(r,r')+d_g(r',q)\leq d_g(r,q)$.
\end{itemize}
We have then the following implications:
\begin{itemize}
\item[] (b1)+(c1)+(c2) $\Then$ (d1)\quad $d_g(p,r')\geq e^\lambda d_g(r',q)+(1+e^\lambda)d_g(r,r')$;
\item[] (b2)+(a1)+(a2) $\Then$ (d2)\quad $d_g(p,r')=e^\lambda d_g(r',q)+(1+e^\lambda)d_g(r_\lambda,r')$;
\item[] (d1)+(d2) $\Then$ (d3)\quad $d_g(r_\lambda,r')\geq d_g(r,r')$.
\end{itemize}
However, from the definition of $r$ we must have $d_g(r_\lambda,r')\leq d_g(r,r')$, hence
the claim follows.
\end{proof}

\begin{corollary}\label{c2}
Any maximal, unit-speed future directed timelike geodesic segment $\gamma:[0,d_g(p,q)]$ 
between $p\ll q$ is an orbit of $\lambda^g_{p,q}$, up to reparametrization.
\end{corollary}
\begin{proof}
Notice that, due to the maximality of $\gamma$, $\gamma(t)$ doesn't belong to the timelike 
cut locus of neither $p$ nor $q$, hence $d_g(p,\cdot)$ and $d_g(\cdot,q)$ are smooth in 
an open neighborhood of $\gamma(t)$ for any $t\in(0,d_g(p,q))$, and thus so are the level 
sets of $\lambda^g_{p,q}$. Moreover, maximality of $\gamma$ together with Proposition \ref{p5} 
also imply that the direction of the tangent vector $\dot{\gamma}(t)$ coincides with the 
direction of greatest variation of the ratio $\frac{d_g(p,\cdot)}{d_g(\cdot,q)}$. Since $\log$ 
is strictly increasing, this shows that the covector $g(\dot{\gamma}(t),\cdot)(\gamma(t))$ 
points in the same direction as $\ud\lambda^g_{p,q}(\gamma(t))$, thus establishing the first 
part of the claim. Since $d_g(\gamma(t),q)=d_g(p,q)-d_g(p,\gamma(t))$ and $t=d_g(p,\gamma(t))$ 
by maximality and the unit parametrization of $\gamma$, the desired reparametrization of 
$\gamma$ is given by 
\[
t\mapsto\lambda^g_{p,q}(\gamma(t))=\log\left(\frac{t}{d_g(p,q)-t}\right).
\]
\end{proof}

Proposition \ref{p5} together with formula (d2) in its proof gives two interesting 
alternative expressions for $\lambda^g_{p,q}$:

\begin{equation}\label{e9}
\lambda^g_{p,q}(r)=\log\left(\frac{d_g(p,r)}{d_g(r,q)}\right)=\left\{\begin{array}{lr}-\log
\left(e^\lambda+(1+e^\lambda)\frac{d_g(r,(\lambda^g_{p,q})^{-1}(\lambda))}{d_g(p,r)}\right) & 
(\lambda^g_{p,q}(r)\leq\lambda)\\ \log\left(e^\lambda+(1+e^\lambda)\frac{d_g((\lambda^g_{p,q})^{-1}
(\lambda),r)}{d_g(r,q)}\right) & (\lambda^g_{p,q}(r)\geq\lambda)\end{array}\right.
\end{equation}

We will use this formula in Section \ref{sec:4} to give a generalization of 
$\lambda^g_{p,q}$ which is defined when $p$ (resp. $q$) is not given, but instead
an acausal hypersurface to the past of $q$ (resp. future of $p$) is. 


\subsection{Finer details in the small}\label{sec:3:small}

If $\OO_{p,q}$ is suitably small, then much more information about $\lambda^q_{p,q}$ can be 
obtained. To this avail, recall now that, for $(\M,g)$ \emph{strongly causal} and $p\ll q$ 
belonging to \emph{geodesically convex} regions, $\frac{1}{2}d_{\bar{g}}^2$ coincides within 
$\{(r,s):p\ll r\leq s\ll q\}$ with \emph{Synge's world-function}, given by 
\[
\Gamma_g(p,q)\doteq-\frac{1}{2}\int^1_0\bar{g}(\dot{\gamma}_{p,q}(s),\dot{\gamma}_{p,q}(s))ds,
\] 
with $\gamma_{p,q}$ the (only) geodesic segment from $p=\gamma_{p,q}(0)$ to $q=\gamma_{p,q}(1)$ 
\cite{beemee}. We recall from \cite{fried} the following properties of $\Gamma_g(p,q)$:

\begin{itemize}
\item $\Gamma_g\in\mathscr{C}^\infty\left(\bigcup_{p\in\mathscr{M}}\{p\}\times\mathscr{U}_p\right)$, 
where $\mathscr{U}_p$ is an open, geodesically convex neighborhood of $p$;
\item $\Gamma_g(p,q)=\Gamma_g(q,p)$;
\item $\nabla^a\Gamma_g(p,\cdot)=-\dot{\gamma}^a_{p,.}(.)$ and $\nabla^a\Gamma_g(\cdot,q)=
-\dot{\gamma}^a_{.,q}(\cdot)$, where $\cdot$ denotes the variable on which $\nabla$ acts. We 
immediately have the fundamental Gauss's Lemma
\begin{equation}\label{e10}
g^{-1}(\ud_p\Gamma_g(p,q),\ud_p\Gamma_g(p,q))=g^{-1}(\ud_q\Gamma_g(p,q),\ud_q\Gamma_g(p,q))=-2\Gamma_g(p,q),
\end{equation}
where $\ud_p$ and $\ud_q$ denote respectively the differential with respect to the first and second 
variables.
\item $(\nabla_a\nabla_b\Gamma_g(p,.))(p)=-g_{ab}(p)$.
\end{itemize}

Recalling that any causally simple spacetime is strongly causal, we write $\lambda^{g}_{p,q}(r)
\doteq\frac{1}{2}(\log(\Gamma_g(p,r))-\log(\Gamma_g(r,q)))$, whence it follows that 
\begin{equation}\label{e11}
T^a=\frac{\nabla^a\lambda^g_{p,q}}{g^{-1}(d\lambda^g_{p,q},d\lambda^g_{p,q})}=\frac{\Gamma_g(p,r)
\nabla^a\Gamma_g(r,q)-\Gamma_g(r,q)\nabla^a\Gamma_g(p,r)}{\Gamma_g(p,r)+\Gamma_g(r,q)+g^{-1}(d
\Gamma_g(p,r),d\Gamma_g(r,q))}
\end{equation}
generates the flux of diffeomorphisms $\lambda\mapsto u^\lambda_{p,q}$ associated to 
the foliation induced by $\lambda^g_{p,q}$, that is, $T$ is the unique vector field 
satisfying the following properties:

\begin{itemize}
\item $\ud\lambda^g_{p,q}(r)(T)=1$ for all $r\in\OO_{p,q}$;
\item $g(T,X)=0$ for all $X^a$ tangent to $(\lambda^g_{p,q})^{-1}(t)$, $t\in\RR$.
\end{itemize}

To obtain formula \eqref{e11}, notice that the above properties imply that
\begin{equation}\label{e12}
g(T,T)=\frac{1}{g^{-1}(\ud\lambda^g_{p,q},\ud\lambda^g_{p,q})},
\end{equation}
where
\begin{equation}\label{e13}
g^{-1}(\ud\lambda^g_{p,q},\ud\lambda^g_{p,q})=-\frac{1}{2}\left(\frac{1}{\Gamma_g(p,\cdot)}+\frac{1}{
\Gamma_g(\cdot,q)}+\frac{g^{-1}(\ud\Gamma_g(p,\cdot),\ud\Gamma_g(\cdot,q))}{\Gamma_g(p,\cdot)\Gamma_g
(\cdot,q)}\right),
\end{equation}
of which \eqref{e11} is an immediate consequence. Let us now pay attention to the limiting form 
of $T$ in the future (resp. past) horizon $\partial_+\OO_{p,q}$ (resp. $\partial_-\OO_{p,q}$). 
We obtain these limits by sending $\Gamma_g(r,q)\To 0$ (resp. $\Gamma_g(p,r)\To 0$) in \eqref{e9} 
while keeping $\Gamma_g(p,r)$ (resp. $\Gamma_g(r,q)$) constant (notation: $\lim_{\To\partial_\pm\OO_{p,q}}$), 
yielding
\begin{equation}\label{e14}
\begin{split}
\lim_{\To\partial_+\OO_{p,q}}T^a &=-\frac{\Gamma_g(\cdot,q)\nabla^a\Gamma_g(p,\cdot)}{\Gamma_g(\cdot,q)+g^{-1}
(\ud\Gamma_g(p,\cdot),\ud\Gamma_g(\cdot,q))}\\
\lim_{\To\partial_-\OO_{p,q}}T^a &=\frac{\Gamma_g(p,\cdot)
\nabla^a\Gamma_g(\cdot,q)}{\Gamma_g(p,\cdot)+g^{-1}(\ud\Gamma_g(p,\cdot),\ud\Gamma_g(\cdot,q))}
\end{split}
\end{equation}
That is, $T$ extends continuously to a \emph{null} vector field over $\partial\OO_{p,q}$ which is tangent 
and normal to $\partial_{+}\OO_{p,q}$ and $\partial_-\OO_{p,q}$, and vanishing at $p$, $q$ and $\E_{p,q}$. 
Particularly, $T$ is tangent to the null generators of $\partial_{+}\OO_{p,q}$ and $\partial_-\OO_{p,q}$,
which entails that $\nabla_TT=\kappa^\pm_{p,q}T$, where $\kappa^+_{p,q}$ (resp. $\kappa^-_{p,q}$) is a 
scalar function on $\partial_{+}\OO_{p,q}$ (resp. $\partial_-\OO_{p,q}$) which measures the failure of 
any extension of $\lambda^g_{p,q}$ to $\partial_{+}\OO_{p,q}$ (resp. $\partial_-\OO_{p,q}$) in being 
an affine parameter for the latter's null generators.\footnote{We remark, however, that $T$ determines
this extension up to a smooth choice of an additive constant on each null generator, in the same
way as the extension of isometries in Minkowski spacetime to null infinity in the sense of Penrose
suffers from the so-called supertranslation ambiguity \cite{wald}.} Another way of seeing 
$\kappa^\pm_{p,q}$ is as the magnitude (up to a sign) of the ``near-horizon'' acceleration 
$g(T,T)^{-1}\nabla_T T$ of the orbits of $T$ multiplied by the redshift factor $(-g(T,T))^{\frac{1}{2}}$. 
As such, it is fair to call $\kappa^+_{p,q}$ (resp. $\kappa^-_{p,q}$) the \emph{future} (resp. \emph{past})
\emph{surface gravity} of $\OO_{p,q}$, in the spirit of the zeroth law of black hole dynamics \cite{wald}.

To compute the tangential acceleration of $T$ and, therefore, $\kappa^\pm_{p,q}$, we define the auxiliary functions
\begin{align*}
h_{p,q}(r) &\doteq\Gamma_g(p,r)+\Gamma_g(r,q)+g^{-1}(\ud\Gamma_g(p,r),\ud\Gamma_g(r,q)),\\
f_p(r) &\doteq\frac{\Gamma_g(p,r)}{h_{p,q}(r)},\\
f_q(r) &\doteq\frac{\Gamma_g(r,q)}{h_{p,q}(r)},
\end{align*}
which allows us to write
\begin{equation}\label{e15}
\begin{split}
T^a &=\frac{1}{2}(f_q\nabla^a\Gamma_g(p,r)-f_p\nabla^a\Gamma_g(r,q)),\\
g(T,T) &= -2\frac{\Gamma_g(p,\cdot)\Gamma_g(\cdot,q)}{h_{p,q}}\\
 &=-2f_p\Gamma_g(\cdot,q)\\
 &=-2f_q\Gamma_g(p,\cdot)\\
 &=-2h_{p,q}f_pf_q.
\end{split}
\end{equation}
This gives us
\begin{equation}\label{e16}
\begin{split}
\nabla_{a}T_b =&\frac{1}{2}\left(f_q\nabla_{a}\nabla_b\Gamma_g(p,\cdot)-f_p\nabla_{a}\nabla_b\Gamma_g(\cdot,q)\right)\\
 &+\frac{1}{2}\left(\nabla_b\Gamma_g(p,\cdot)\nabla_{a}f_q-\nabla_b\Gamma_g(\cdot,q)\nabla_{a}f_p\right)\\
 =&\frac{1}{2}\left(f_q\nabla_{a}\nabla_b\Gamma_g(p,\cdot)-f_p\nabla_{a}\nabla_b\Gamma_g(\cdot,q)\right)\\
 &+\frac{1}{2h_{p,q}}\left(\nabla_{a}\Gamma_g(\cdot,q)\nabla_b\Gamma_g(p,\cdot)-\nabla_b\Gamma_g(\cdot,q)\nabla_{a}\Gamma_g(p,\cdot)\right)\\
 &-\frac{1}{h_{p,q}}T_b\nabla_{a}h_{p,q}.\\
\end{split}
\end{equation}
The first term is manifestly symmetric, whereas the second is manifestly antisymmetric.
Hence
\begin{equation}\label{e17}
\begin{split}
T^{a}\nabla_{a}T_b =&\frac{1}{4}\left(f_q\nabla^a\Gamma_g(p,r)-f_p\nabla^a\Gamma_g(r,q)\right)
\left(f_q\nabla_{a}\nabla_b\Gamma_g(p,\cdot)-f_p\nabla_{a}\nabla_b\Gamma_g(\cdot,q)\right)\\
 &+\frac{1}{4h_{p,q}}\left(f_q\nabla^a\Gamma_g(p,r)-f_p\nabla^a\Gamma_g(r,q)\right)
\left(\nabla_{a}\Gamma_g(\cdot,q)\nabla_b\Gamma_g(p,\cdot)\right.\\
 &\left.-\nabla_b\Gamma_g(\cdot,q)\nabla_{a}\Gamma_g(p,\cdot)\right)
+\frac{1}{h_{p,q}}T^bT^{a}\nabla_{a}h_{p,q}\\
 =& -\frac{1}{2}\left(f^2_q\nabla_b\Gamma_g(p,\cdot)+f^2_p\nabla_b\Gamma_g(\cdot,q)
+f_pf_q\nabla_bg^{-1}(\ud\Gamma_g(p,\cdot),\ud\Gamma_g(\cdot,q))\right)\\
 &+\frac{1}{4h_{p,q}}\left[\left(f_q\nabla_b\Gamma_g(p,\cdot)+f_p\nabla_b\Gamma_g(\cdot,q)\right)
g^{-1}(\ud\Gamma_g(p,\cdot),\ud\Gamma_g(\cdot,q))\right.\\
 &\left.+2\left(f_q\Gamma_g(p,\cdot)\nabla_b\Gamma_g(\cdot,q)
-f_p\Gamma_g(\cdot,q)\nabla_b\Gamma_g(p,\cdot)\right)\right]+\frac{1}{h_{p,q}}T^bT^{a}\nabla_{a}h_{p,q},
\end{split}
\end{equation}
where we have used formula \eqref{e10} in the last passage, and, recalling that the second term
is antisymmetric,
\begin{equation}\label{e18}
\begin{split}
T^{a}T^b\nabla_{a}T_b =&\frac{1}{2}T^{a}\nabla_{a}g(T,T)\\
 =&-\frac{1}{4}\left(f_q\nabla^b\Gamma_g(p,r)-f_p\nabla^b\Gamma_g(r,q)\right)
\left(f^2_q\nabla_b\Gamma_g(p,\cdot)+f^2_p\nabla_b\Gamma_g(\cdot,q)\right.\\ 
 &\left.+f_pf_q\nabla_bg^{-1}(\ud\Gamma_g(p,\cdot),\ud\Gamma_g(\cdot,q))\right)
+\frac{g(T,T)}{h_{p,q}}T^{a}\nabla_{a}h_{p,q}\\
 =&-\frac{g(T,T)}{4}\left(f^2_q+f^2_p-(f_p-f_q)\frac{1}{2h_{p,q}}
g^{-1}(\ud\Gamma_g(p,\cdot),\ud\Gamma_g(\cdot,q))\right.\\
 &\left.-\frac{1}{2h_{p,q}}T^b\nabla_bg^{-1}(\ud\Gamma_g(p,\cdot),\ud\Gamma_g(\cdot,q))\right)
+\frac{g(T,T)}{h_{p,q}}T^{a}\nabla_{a}h_{p,q}.
\end{split}
\end{equation}

To compute 
\[
\kappa^\pm_{p,q}=\lim_{\To\partial_\pm\OO_{p,q}}\frac{T^{a}T^b\nabla_{a}T_b}{g(T,T)},
\]
we need to know the limiting form of $h_{p,q}(r)$, $f_p(r)$ and $f_q(r)$ as $r$ approaches
$\partial\OO_{p,q}$. In fact, we have that
\[
h_{p,q}\left\{\begin{array}{rl}\stackrel{\To\partial_+\OO_{p,q}}{\longrightarrow} & \Gamma_g(p,\cdot)+g^{-1}(\ud\Gamma_g(p,\cdot),\ud\Gamma_g(\cdot,q))\doteq h^+_{p,q}\\
\stackrel{\To\partial_-\OO_{p,q}}{\longrightarrow} & \Gamma_g(\cdot,q)+g^{-1}(\ud\Gamma_g(p,\cdot),\ud\Gamma_g(\cdot,q))\doteq h^-_{p,q}\end{array}\right..
\]
Therefore 
\[
\begin{split}
f_p &\left\{\begin{array}{rl}\stackrel{\To\partial_+\OO_{p,q}}{\longrightarrow} & \frac{\Gamma_g(p,\cdot)}{h^+_{p,q}}\doteq f^+_p\\
\stackrel{\To\partial_-\OO_{p,q}}{\longrightarrow} & 0\end{array}\right.,\\
f_q &\left\{\begin{array}{rl}\stackrel{\To\partial_+\OO_{p,q}}{\longrightarrow} & 0\\
\stackrel{\To\partial_-\OO_{p,q}}{\longrightarrow} & \frac{\Gamma_g(\cdot,q)}{h^-_{p,q}}\doteq f^-_q\end{array}\right.,
\end{split}
\]
whence we conclude that
\[
\begin{split}
\kappa^+_{p,q} =&-\frac{1}{4}\left[\frac{f^+_p}{h^+_{p,q}}\left(\Gamma_g(p,\cdot)-\frac{1}{2}g^{-1}(\ud\Gamma_g(p,\cdot),\ud\Gamma_g(\cdot,q))\right)\right.\\
 &\left.-\frac{1}{2h^+_{p,q}}T^b\nabla_bg^{-1}(\ud\Gamma_g(p,\cdot),\ud\Gamma_g(\cdot,q))\right]+\frac{1}{h^+_{p,q}}T^{a}\nabla_{a}h^+_{p,q},\\
\kappa^+_{p,q} =&-\frac{1}{4}\left[\frac{f^-_q}{h^-_{p,q}}\left(\Gamma_g(\cdot,q)+\frac{1}{2}g^{-1}(\ud\Gamma_g(p,\cdot),\ud\Gamma_g(\cdot,q))\right)\right.\\
 &\left.-\frac{1}{2h^-_{p,q}}T^b\nabla_bg^{-1}(\ud\Gamma_g(p,\cdot),\ud\Gamma_g(\cdot,q))\right]+\frac{1}{h^-_{p,q}}T^{a}\nabla_{a}h^-_{p,q}.
\end{split}
\]


\begin{remark}\label{r2}
There is a hypothesis about the normalization of $T^a$ implicit in the
definition of $\kappa^\pm_{p,q}$. More precisely, our definition is conditioned 
to the following fact: if $r$ is the middle point of the maximal timelike geodesic 
in $\M$ linking $p$ to $q$ (and hence $\Gamma_g(p,r)=\Gamma_g(r,q)$), clearly 
implying that $\nabla^a\Gamma_g(p,r)=-\nabla^a\Gamma_g(r,q)$, then $g(T,T)(r)=
-\frac{d_g(p,q)^2}{16}$. If we rescale $T^a$ by a factor $R\in\RR$, $\kappa^\pm_{p,q}$ 
is rescaled by the same factor, by definition. In Minkowski spacetime, for instance, 
this has the consequence that, if we take $p=p_0$ and $q=q_0$ as in Subsection 
\ref{sec:2:mink} and rescale the $T^a$ associated to the diamond $\OO_{Rp_0,Rq_0}$ 
by a factor $\frac{2}{R}$, so as to maintain $g(T,T)(r)=\eta(T,T)(0)=-1$, it follows 
that $\kappa^\pm_{p_0,q_0}\gotoas{0}{R}{+\infty}$, in a way consistent with the fact that 
$\frac{2}{R}T^a\gotoas{(\partial_0)^a}{R}{+\infty}$ (physically, the ``Unruh temperature'' 
associated to time translations in Minkowski spacetime is zero).
\end{remark}

\section{A variation over the theme: half-diamonds}\label{sec:4}

In this section, we will apply our strategy to a slightly different kind of region. Consider
a causally simple spacetime $(\M,g)$ endowed with a smooth global time function $\tau$, and
denote the latter's level hypersurfaces by $\Sigma_t\doteq\tau^{-1}(t)$. 

\begin{definition}\label{d2}
Let $t\in\RR$, $p\in I^-\Sigma_t$, $q\in I^+\Sigma_t$ such that $J^+(p)\cap J^-(\Sigma_t)$ 
and $J^-(q)\cap J^+(\Sigma_t)$ are compact. The (relatively compact) regions of the form 
$\OO^-_p(t)\doteq I^+(p)\cap I^-(\Sigma_t)$ and $\OO^+_q(t)\doteq I^-(q)\cap I^+(\Sigma_t)$ 
are respectively called \emph{past} and \emph{future half-diamonds} at time $t$.
\end{definition}

Half-diamonds also enjoy the property of being globally hyperbolic, but, similarly to the
case of diamonds, generally are \emph{not} of the form $int D^-(J^+(p)\cap\Sigma_t)$ or
$int D^+(J^-(q)\cap\Sigma_t)$ with respect to the ambient spacetime, unless they are 
contained in a geodesically convex neighborhood.

The alternative formulae for $\lambda^g_{p,q}$ derived in Section 3 suggest the following
intrinsic global time functions for $\OO^-_p(t)$ and $\OO^+_q(t)$:
\begin{align}
\OO^-_p(t)\ni r & \mapsto \lambda^{g,t}_{-,p}(r)\doteq-\log\left(e^t+(1+e^t)\frac{d_g
(r,\Sigma_t)}{d_g(p,r)}\right)\in(-\infty,t];\label{e19}\\
\OO^+_q(t)\ni r & \mapsto \lambda^{g,t}_{+,q}(r)\doteq\log\left(e^t+(1+e^t)\frac{d_g
(\Sigma_t,r)}{d_g(r,q)}\right)\in[t,+\infty).\label{e20}
\end{align}

Notice that the definition of $\lambda^{g,t}_{-,p}$ (resp. $\lambda^{g,t}_{+,q}$) doesn't use any 
information whatsoever about $\Sigma_t\setminus J^+(p)$ (resp. $\Sigma_t\setminus J^-(q)$) -- 
particularly, if $p\ll q$ and $\Sigma_t\cap I^+(p)$ (resp. $\Sigma_t\cap I^-(q)$) equals 
$\lambda^g_{p,q}{}^{-1}(0)$, then $\lambda^{g,t}_{-,p}=\lambda^g_{p,q}$ (resp. $\lambda^{g,t}_{+,q}=
\lambda^g_{p,q}$) in $\OO^-_p(t)$ (resp. $\OO^+_q(t)$). 

Employing the same methods used in our study of $\lambda^g_{p,q}$, one concludes that
$\lambda^{g,t}_{-,p}$ and $\lambda^{g,t}_{+,q}$ are locally Lipschitz and twice differentiable 
almost everywhere, and foliate respectively $\OO^-_p(t)$ and $\OO^+_q(t)$ by acausal Cauchy 
surfaces. Moreover, for $\OO^-_p(t)$, $\OO^+_q(t)$ contained in a geodesic neighborhood, one 
can derive analogous formulae for the vector fields generating the diffeomorphism flows 
associated to these foliations -- particularly, though it is not immediately obvious from the 
formulae above, the vector fields generating the flows do admit smooth extensions respectively 
to the past horizon $\partial^-\OO^-_p(t)\doteq\partial I^+(p)\cap J^-(\Sigma_t)$ of $\OO^-_p(t)$
and the future horizon $\partial^+\OO^+_q(t)\doteq\partial I^-(q)\cap J^+(\Sigma_t)$ of 
$\OO^+_q(t)$, being there once more tangent to the null geodesic generators. Moreover, 
the asymptotic surface gravities near $p$ and $q$ are the same as in the case of diamonds.


%

\section{Conclusions and remarks (or: a quantum coda)}\label{sec:5}

We have presented a rather general procedure of building global time functions for
diamonds in general spacetimes. Our assumptions seem to be optimal. The method is 
sufficiently robust to be adapted so as to satisfy boundary conditions at conformal 
infinity, as for instance in the case of wedges in spacetimes endowed with past and 
future null infinity, in the spirit of \cite{ribeiro1,ribeiro2} (to be addressed in a 
future publication).

One must mention some potential applications for \emph{quantum field theory in curved 
spacetimes}, which to a great extent inspired the developments presented above: it is 
known \cite{hislongo,bglongo} that the vacuum state $\omega_0$ of a conformally 
invariant quantum field theory in Minkowski spacetime is a \emph{KMS} (i.e., 
\emph{finite-temperature}) \emph{state} w.r.t. the W*-dynamical system $(\A(\OO)_{p,q}),
\alpha_\lambda)$, where $\A(\OO_{p,q})$ is the local von Neumann algebra of observables 
associated to $\OO_{p,q}$ and $\alpha_\lambda$ is the group of *-automorphisms induced 
by $u^\lambda_{p,q}$, and hence $-2\pi\lambda$ is precisely the flow parameter of the 
\emph{Tomita-Takesaki modular group} intrinsically associated to the standard pair 
$(\A(\OO_{p,q}),\omega_0)$. This result led Martinetti and Rovelli \cite{martrov} to 
propose that $\lambda$ should then be realized as a ``thermal time'' for finite 
lifetime observers in curved spacetime. Heuristically, in the general situation we have
discussed, an asymptotic form of the KMS periodicity condition as $\lambda\To\pm\infty$
(i.e. a kind of ``return to equilibrium'') is strongly suggested by leading term of 
the asymptotic short-distance expansion of the two-point function in powers of Synge's 
world function \cite{riesz,fried,bargp}. An honest proof of this fact, however, requires 
obtaining certain decay estimates for solutions of the wave and Klein-Gordon equations 
in diamonds at large ``cosmological times'' $\lambda^g_{p,q}\To\pm\infty$. Work on such
estimates is in progress. Having \emph{asymptotic freedom} in mind, one is hence led 
to the possibility that the ``thermal time hypothesis'' is only realized 
\emph{asymptotically} as a \emph{return to equilibrium} for more realistic QFT's, since 
this becomes essentially (at least, on a geometrical level) a scaling limit.

%

\section*{Acknowledgments}

A large part of the material presented here is taken from Chapters 1 and 2 of my
PhD thesis \cite{ribeiro2}, which was supported by FAPESP under grant no. 01/14360-1. 
I would like to thank the II. Institut f\"{u}r theoretische Physik, Universit\"{a}t Hamburg 
for the hospitality during the long and somewhat disconnected process of writing 
this paper and improving the original results from my thesis.







\end{document}